\documentclass[11pt]{article}
\usepackage{graphicx} 
\usepackage{caption}
\usepackage{amsmath}
\usepackage[T1]{fontenc}
\usepackage[utf8]{inputenc}
\usepackage{authblk}
\usepackage{hyperref}
\usepackage{bbm}
\usepackage{amsthm}
\usepackage{multirow}
\DeclareMathOperator*{\argmin}{argmin} 
\usepackage{float}
\usepackage{mathtools}
\usepackage{amsfonts}
\usepackage{nicematrix}

\DeclarePairedDelimiter\floor{\lfloor}{\rfloor}
\newtheorem{theorem}{Theorem}

\RequirePackage[authoryear]{natbib}

\title{An Encoding Approach for Stable Change Point Detection}
\author[1]{Xiaodong Wang}
\author[1]{Fushing Hsieh}
\affil[1]{Department of Statistics, University of California, Davis.}
\date{} 

\begin{document}

\maketitle

\section*{Abstract}

Without imposing prior distributional knowledge underlying multivariate time series of interest, we propose a nonparametric change-point detection approach to estimate the number of change points and their locations along the temporal axis. We develop a structural subsampling procedure such that the observations are encoded into multiple sequences of Bernoulli variables. A maximum likelihood approach in conjunction with a newly developed searching algorithm is implemented to detect change points on each Bernoulli process separately. Then, aggregation statistics are proposed to collectively synthesize change-point results from all individual univariate time series into consistent and stable location estimations. We also study a weighting strategy to measure the degree of relevance for different subsampled groups. Simulation studies are conducted and shown that the proposed change-point methodology for multivariate time series has favorable performance comparing with currently popular nonparametric methods under various settings with different degrees of complexity. Real data analyses are finally performed on categorical, ordinal, and continuous time series taken from fields of genetics, climate, and finance.


\section{Introduction}

Upon any nonstationary time series of any dimensions, abrupt distributional changes are ubiquitous patterns of great interest found in many sciences involving with evolving systems. Change points as temporal locations of such occurrences and their multiplicity are key parts of deterministic structures of the time series under study. Nowadays, change-point analysis has well recognized in statistics literature and beyond as an essential scientific methodology that aims to detect change points on the time-ordered observations and then partition the whole time series into homogeneously distributional segments. Change-point analysis can be traced back to 1950s (\cite{Page}; \cite{Chernoff and Zacks}; \cite{Kander and Zacks}). So far, it has been playing a crucial role in diverse fields including bioinformatics (\cite{Picard}; \cite{Muggeo and Adelfio}), behavioural science(\cite{Rosenfield}; \cite{Hoover}), neuroimage (\cite{Bosc}), climate science (\cite{Robbins}), finance (\cite{Talih and Hengartner}), and speech recognition (\cite{Malladi}).

In general, such an analysis can be conducted under either parametric or nonparametric settings. Parametric approaches rely heavily upon assumptions of underlying distributions belonging to a known family. Likelihood or penalized likelihood functions are generally involved (\cite{Yao}; \cite{Chen and Gupta}; \cite{Bai and Perron}). In contrast, nonparametric approaches make very few assumptions regarding stochasticity underlying the time series. The likelihood principle is not directly applicable. Nonetheless, such approaches fit well in a wider variety of applications. Such an advantageous feature has popularity and a vast amount of research attention in the past decade.

In fact, the likelihood principle is still applicable under a rather mild independence assumption that, at least, approximately endorses some distributional characters upon observed or computed recurrent events occurring along with the time series. For instance, (\cite{Kawahara and Sugiyama}; \cite{Liu}) attempted to estimate the likelihood ratio using KL divergence; (\cite{Chen and Zhang}) proposed a graph-based approach and applied it in multivariate non-Euclidean data. (\cite{Zou}) developed an empirical likelihood approach to discover an unknown number of change-points via BIC. (\cite{Matteson and James}) present a U-statistic to quantify the difference between the characteristic functions of two segments. (\cite{Lung-Yut-Fong}) generalized Mann-Whitney rank-based statistic to multivariate settings. (\cite{Arlot}) improved the kernel-based method by (\cite{Harchaoui and Cappe}) with a generalized model-selection penalty. 

However, most of the existing nonparametric research focused on the single change-point problem and the extension of multiple change point detection is achieved via dynamic programming (\cite{Harchaoui and Cappe}; \cite{Lung-Yut-Fong}; \cite{Arlot}) or bisection procedure (\cite{Vostrikova}; \cite{Olshen and Venkatraman}; \cite{Matteson and James}). It is still scarce in the literature to efficiently discover multiple change points under multivariate settings, especially when the covariance structure changes in chronological order.

In the paper, a new nonparametric approach is proposed to detect multiple distributional changes. Our developments are anchored on independent time-ordered observations. The basic idea is to systematically select a subset of the data points at each iteration, with which we encode the continuous observations into a sequence of Bernoulli variables. The number of change points and their locations are estimated by aggregating all the dynamic information discovered from the collection of Bernoulli processes. Instead of working on the unknown distribution directly, the proposed approach takes advantage of dividing the problem into several easier tasks, so that the maximum likelihood approach can be applied to analyze the Bernoulli processes, respectively. We demonstrate that this divide-and-conquer framework is robust to any underlying distributions and can be implemented in conjunction with other parametric approaches.

Another important extension of the aggregation technique is the stability change-point detection. Such a stability selection introduced by (\cite{Meinshausen and Buhlmann}) was designed to improve the performance of variable selection and provide control for false discoveries. We demonstrate that the idea of aggregating results by applying a procedure to subsamples of the data can be well implemented under our framework. One can aggregate the estimation from the Bernoulli sequences, and select the estimated change-point locations with votes beyond a predetermined threshold. To our limited knowledge, this could be the first method in the change-point literature that holds both asymptotic property and finite-sample control of false discoveries.

The paper is organized as follows. We start with an efficient algorithm for searching multiple change points within a change-in-parameter Bernoulli sequence in Section~2. In Section~3, we propose the main divide-and-conquer framework to analyze multivariate observations. In Section~4, the stability detection technique is applied under our change point framework. In Section~5, a strategy is provided to weighting the results from different sample sets for practical usage. Numerical experiments are shown in Section~6 to compare with other well-known nonparametric approaches. Real data applications including categorical and continuous data in univariate and multivariate settings are reported in Section~7. We note that the proposed approach can be easily generalized to multivariate categorical or ordinal time series data, though we mainly focus on continuous data under the multivariate setting in this paper.

\section{Sequence of Bernoulli variables}

\subsection{Background}

Consider a sequence of 0-1 independent Bernoulli variables $\{E_t\}_{t=1}^N$. Suppose that $k$ change points are embedded within the sequence at locations $0 = \tau^*_0 < \tau^*_1 < ... < \tau^*_k < \tau^*_{k+1}=N$, so the observations are partitioned into $k+1$ segments. Observations within segments are identically distributed but observations between adjacent segments are not. Specially, $E_t \stackrel{iid}{\sim} Bern(p_i)$ for $E_t \in \{E_{\tau^*_{i}+1}, ...,E_{\tau^*_{i+1}}\}$, for $i=0,...,k$. Now, given the number of change points $k$, one task of change point detection is to estimate the $k$ locations. In the most general case, both number of change points and their locations need to be estimated. 

Change point analysis in a Bernoulli-variable sequence was well studied when $k=1$. (\cite{Hinkley and Hinkley}) provided asymptotic distributions of likelihood ratio statistics for testing the existence of a change point. (\cite{Pettitt}) introduced CUSUM statistics and showed its asymptotical equivalence to the maximum likelihood estimator. (\cite{Miller and Siegmund}) investigated maximally selected chi-square statistics for two-sample comparison in a form of 2$\times$2 table. Later on, (\cite{Halpern}) advocated a statistic based on the minimum value of Fisher's Exact Test. When $k>1$, (\cite{Fu and Curnow}) firstly attempted to search for optimal change points such that the likelihood function is maximized. However, it still lacks a computationally efficient algorithm especially when $k$ is large.

In this section, we present a new algorithm to address the problem of performing multiple change points detection within a Bernoulli-variable sequence. An exhaustive searching procedure is proposed but with relatively feasible time complexity. The idea is motivated by (HFS \cite{Hsieh}) which was designed to detect dynamics phase shifts from one episode to another in financial data. By tracking the recurrence of 1's in the time axis, observations are partitioned into disjoint segments with different emergence intensities in a fashion of dynamic programming. Thus, change points between adjacent segments are detected such that the likelihood or penalized-likelihood functions are maximized.


\subsection{Multiple Change Points Searching Algorithm}

For simplicity of computation, we only consider the situation in which change point locates at the emergence position of 1's. Suppose that the number of 1's in the $i$-th segment is $M_i$, so the total number of 1's is $M=\sum_{i=1}^{k+1} M_i$ and the total number of 0's is $N-M$. By further supposing that the recurrent time can be 0 if two 1's appear consecutively, and $R_1=0$ if $E_1=1$, and $R_{M+1}=0$ if $E_N=1$, the Bernoulli-variable sequence can be represent by a sequence of recurrent time between consecutive 1's, denoted as $\{R_t\}_{t=1}^{M+1}$. Especially, there are $M_i+1$ recurrent times in the $i$-th segment where $R_t \sim Geom(p_i)$. The task then becomes to search for the change points within the recurrent-time sequence.

The searching procedure is done by iteratively taking off the smallest number $R_{min}$ from the rest $R_t$'s and combine the time points within $R_{min}$. For example, if $R_{min}$ is the recurrent time between $j$ and $j'$, we combine the locations from $(j+1)$ to $j'$ as a time window, denoted as $w_{j\to j'}$. Here, we suppose that $E_j=1$, $E_{j'}=1$, and $E_t=0$ for $t \in (j,j')$. In the next step, if the smallest $R_t$ is taken from the recurrent time between $j'$ and $j''$, a new time window is recorded from $(j'+1)$ to $j''$, named $w_{(j'+1)\to j''}$. We can further combine the two consecutive time windows $w_{(j+1)\to j'}$ and $w_{(j'+1)\to j''}$ into $w_{(j+1)\to j''}$. Indeed, we iteratively merge a pair of nearest 1's at each step and update the recorded time windows according to their connectivity. The recorded time windows contain recurrent time with relatively smaller values, which corresponds to a period with high frequency of 1's. Hence, the boundaries of the time windows can be extracted as potential change point locations that partition the observations into segments with low and high Bernoulli parameters.

So far, the algorithm works very similarly to the hierarchical clustering with a single-linkage, by merging two closest single 1's or two groups from bottom to top. However, it is known that this greedy algorithm does not guarantee global optimization. Our remedy is to set a tuning parameter $C^*$ to control the minimal length of the recorded high-intensity segments. Additionally, we count the number of $R_t$ absorbed within each recorded time window. Continuing with the above example, the count of recurrent time for window $w_{(j+1)\to j'}$ and $w_{(j+1)\to j''}$ is denoted as $C_{(j+1)\to j'}=1$ and $C_{(j+1)\to j''}=2$, respectively. The recorded time window, for example, $w_{.\to ..}$ is regarded as a high-intensity segment only if its count $C_{.\to ..}$ is greater than the threshold $C^*$. Hierarchical clustering with a single-linkage is just a special case that $C^*=0$. Another most extreme case is when $C^*=M$, so there is no period having a count number above $C^*$, thus no change point exists. Without any prior knowledge about the minimal length of the segments, we run over all the choice of $C^*$ starting from $0$ to $M$ to generate all possible partitions. The optimum is returned to fit the Bernoulli or Geometric observations best. 

Suppose the observations are partitioned into $\Tilde{k}+1$ segments via $\Tilde{k}$ time window boundaries or change points $\Tilde{\tau}_1,...,\Tilde{\tau}_{\Tilde{k}}$. The Bernoulli parameter $\Tilde{p}_i$ between $\Tilde{\tau}_{i-1}$ and $\Tilde{\tau}_{i}$ can be estimated by MLE $\hat{\Tilde{p}}_i=\frac{\{ \#~\text{of}~1's~\in~(\Tilde{\tau}_{i-1},\Tilde{\tau}_i)\}}{\Tilde{\tau}_i-\Tilde{\tau}_{i-1}}$. To measure the goodness-of-fit, model selection is done by maximizing log-likelihood function within each segment, while penalizing the number of change points $k$ and related estimation parameters. The penalized function or loss can be written by,
\begin{equation} \label{eq:loss}
L(\Tilde{\tau}_1,...,\Tilde{\tau}_{\Tilde{k}}) = -2 \sum_{i=1}^{\Tilde{k}+1} \sum_{t \in (\Tilde{\tau}_{i-1},\Tilde{\tau}_i)}[{E_t}log\hat{\Tilde{p}}_i + {(1-E_t)}log(1-\hat{\Tilde{p}}_i)] + \phi(N)Q_{\Tilde{k}}
\end{equation}
where $Q_k$ is the total number parameters; $\phi(N)$ is the penalty coefficient; $\phi(N)=2$ for AIC and $\phi(N)=log(N)$ for BIC.

Suppose that $W(.)$ is a mapping that records the corresponding time window of $R_t$. For example, $W(R_t)=w_{(j+1)\rightarrow j'}$ where $R_t$ is the recurrent time between $j$ and $j'$. It is marked that the segmentation and the loss function can be updated based on the results in the last step. After applying a big loop cycling through $C^*$ from $0$ to $M$, the total time complexity now becomes $O(M^2)$. As a result, an optimal window set is returned, so the change points locations are estimated by their boundaries. The multiple change points searching algorithm is described in \textbf{Algorithm~1}.

\noindent\rule{12.5cm}{0.8pt}\\
\textbf{Algorithm~1}\\
\rule{12.5cm}{0.4pt}\\
Input: unmarked recurrence time $\{R_t\}_{t}$ and a threshold $C^*$\\
Loop: cycle $R_t$ through order statistics $R_{(1)}$, $R_{(2)}$, ..., $R_{(M+1)}$\\
\indent 1. Initial an empty set $\mathbb{W}$ recording the high-intensity time windows \\
\indent 2. Consider 4 ``if'' conditions and obtain a new window $w$,\\
\indent \indent a. If neither $R_{t-1}$ or $R_{t+1}$ is marked:\\
\indent \indent \indent $w=W(R_{t})$\\
\indent \indent b. If $R_{t-1}$ is marked but $R_{t+1}$ is not:\\
\indent \indent \indent merge $W(R_{t-1})$ and $W(R_{t})$ into one window,\\
\indent \indent \indent $w= \{W(R_{t-1})\bigcup W(R_{t})\}$\\
\indent \indent c. If $R_{i-1}$ is not marked but $R_{t+1}$ is:\\
\indent \indent \indent merge $W(R_{t})$ and $W(R_{t+1})$ into one window,\\
\indent \indent \indent $w= \{W(R_{t})\bigcup W(R_{t+1})\}$\\
\indent \indent d. If both $R_{t-1}$ and $R_{t+1}$ are marked:\\
\indent \indent \indent merge $W(R_{t-1})$, $W(R_{t})$, and $W(R_{i+1})$ into one window,\\
\indent \indent \indent $w= \{W(R_{t-1})\bigcup W(R_{t})\bigcup W(R_{i+1})\}$\\
\indent 3. Update the recorded window set $\mathbb{W}$ with $w$ and mark $R_{t}$.\\
\indent 4. If window length $|w|$ is greater than $C^*$:\\
\indent \indent \indent extract the boundaries of windows in $\mathbb{W}$ as $\Tilde{\tau}_1,...,\Tilde{\tau}_{\Tilde{k}}$\\
\indent \indent \indent update loss function $L(\Tilde{\tau}_1,...,\Tilde{\tau}_{\Tilde{k}})$\\
Output: optima boundaries $\hat{\tau}_1,...,\hat{\tau}_{\hat{k}}$\\
\rule{12.5cm}{0.8pt}\\

\section{MCP for multivariate time series}

A large part of change point detection literature deals with continuous observation. In this section, we firstly proposed an encoding approach to categorize continuous time series into multiple Bernoulli sequences, and then analyze change points embedded within the multivariate process. The idea of categorizing real-value observations aims to extract more relevant information and filter out noise. It is claimed that the proposed approach is robust to encode any underlying distributions and is easily generalized to either categorical or continuous observations.

\subsection{Encoding continuous time series}

In the analysis of single stock returns, the authors in (\cite{Hsieh}) utilized a pair of thresholds to mark absolutely large stock returns as 1 and 0 otherwise, then revealed the volatility pattern behind the resultant 0-1 sequence. The encoding process is written as
\begin{equation}\label{eq:cut}
E_t = \begin{cases} 1 & \quad X_t\le \alpha, \,\, X_t\ge \beta\\
0 & \quad Otherwise \end{cases}
\end{equation}
where $\{E_t\}_t$ is an excursion sequence by marking the stock returns. Later, authors in (\cite{2}) proposed an encoding method to explore the local dependence of observations when $X_t \in \mathbb{R}^p$. Following up the idea, we partition $\mathbb{R}^p$ space into $V$ disjoint subarea, denoted as $B^{(v)}$ for $v=1,2,...,V$, and transform the continuous observations $\{X_t\}_{t=1}^{T}$ into $V$ Bernoulli sequences or
a $V$-dimensional multinominal process $\{(E_t^{(1)}, E_t^{(2)}, ..., E_t^{(V)})\}_{t=1}^{T}$, such that
\begin{equation}\label{eq:ball}
E_t^{(j)} = \begin{cases} 1 & \quad X_t\in B^{(j)}\\
0 & \quad Otherwise \end{cases}
\end{equation}
Here, subarea $B^{(j)}$ plays an important role to reserve the change-point pattern into a Bernoulli process. Denote the Bernoulli parameter in the $i$-th segments of $\{E_t^{(j)}\}$ as $p_{i}^{(j)}$. So,
\begin{equation} \label{eq:integ}
p_{i}^{(j)}=\int_{B^{(j)}} dF_{i}
\end{equation}
where $F_{i}$ corresponds to the CDF of $\{X_t\}_t$ in the $i$-th time segments. Consider two consecutive homogeneous time segments $i$ and $i+1$. The change point detection becomes easier if $p_{i}^{(j)}$ is far apart from $p_{i+1}^{(j)}$, and vice versa. There is actually a tradeoff between the size and the total number of the subareas. Larger number of subareas with smaller size can discover the distributional difference more precisely but with sacrifice of the power of statistics due to the reduced sample size. In the following subsections, we would assume that $V$ is fixed and $B^{(j)}$ are determined. The implementation of the encoding procedure is discussed in Section~5.

\subsection{Single Change Point Detection}

Starting with a simplest setting, let's assume that there exists a single change point at $\tau^*$. Specifically, $\{X_t\}_{t=1}^{\tau^{*}} \stackrel{iid}{\sim} F_1$ and $\{X_t\}_{t=\tau^{*}+1}^{N} \stackrel{iid}{\sim} F_2$ where $F_1$ and $F_2$ are two unknown CDFs. The goal is to test the homogeneity between the two sample sets. Following the encoding procedure above, we obtain a multinomial process $\{(E_t^{(1)}, E_t^{(2)}, ..., E_t^{(V)})^{'}\}_{t=1}^{N}$ where $\{E_t^{(j)}\}_{t=1}^{\tau^{*}} \sim Bern(p^{(j)}_{1,\tau^{*}})$ and $\{E_t^{(j)}\}_{t=\tau^{*}+1}^{N} \sim Bern(p^{(j)}_{2,\tau^{*}})$.

(\cite{Robbins}) extent the multivariate CUSUM statistics with uncorrelated components to the multinomial settings and derived its asymptotic distributions under the null hypothesis. The estimators of Bernoulli parameters at a hypothesized time location $\tau$ is defined by
\begin{equation} \label{eq:p1}
    \hat{p}^{(j)}_{1,\tau} = {\sum_{t=1}^{\tau} \mathbbm{1}\{E_t^{(j)}=1\}}/{\tau}
\end{equation}
and 
\begin{equation} \label{eq:p2}
    \hat{p}^{(j)}_{2,\tau} = {\sum_{t=\tau+1}^{N} \mathbbm{1}\{E_t^{(j)}=1\}}/{(N-\tau)}
\end{equation}
for $j=1,2,...,V$. Then, a chi-square statistic proposed by (Robbins et. al 2011) is written as,
\[
\chi_{\tau}^2 = \sum_{j=1}^{V} \frac{(\sum_{t=1}^{\tau}\mathbbm{1}\{E^{(j)}_t=1\}-\hat{p}^{(j)}_{1,\tau})^2}{\hat{p}^{(j)}_{1,\tau}} + \frac{(\sum_{t=\tau+1}^{N}\mathbbm{1}\{E^{(j)}_t=1\}-\hat{p}^{(j)}_{2,\tau})^2}{\hat{p}^{(j)}_{2,\tau}}
\]
Moreover, if there exists no change point under the null hypothesis, the maximally selected chi-square statistics $\chi_{\hat{\tau}}^2$ converges to a Brownian motion asymptotically.


\subsection{Multiple Change Points Detection}

Now, we consider multiple change point detection when number of change point $k$ is known. Suppose the change point locations are $0 = \tau^*_0 < \tau^*_1 < ... < \tau^*_k < \tau^*_{k+1} = N$. Specifically, $\{X_t\}_{t=\tau^*_i}^{\tau^*_{i+1}} \stackrel{iid}{\sim} F_i$ for $i=0,1,...,k$, and consecutive CDFs $F_i$ and $F_{i+1}$ are different. A naive method to search for $O(N^k)$ possible change point locations is computationally intractable. Bisection procedure as in (\cite{Vostrikova}; \cite{Olshen and Venkatraman}), dynamic programming (\cite{Harchaoui and Cappe}), or the one we proposed in Section~2 can work for the purpose. It is claimed that our searching algorithm is favorable in exploring the global optima, but it is designed only adapting to a single-dimensional Bernoulli-variable sequence.

A divide-and-concur approach is proposed as a remedy to the multivariate problem. Denote $\{E_t^{(j)}\}_{t=1}^{N}$ as the $j$-th Bernoulli process after encoding the observations via $B^{(j)}$, and $p^{(j)}_{i}$ as the true parameters of $E_t^{(j)}$ defined by (\ref{eq:integ}). We firstly apply Algorithm1 to estimate the change point locations within $\{E^{(j)}_t\}$, for $j=1,2,...,V$, respectively. Suppose the estimated change point locations in the $j$-th sequence is $0= \hat{\tau}^{(j)}_0 < \hat{\tau}^{(j)}_1 < \hat{\tau}^{(j)}_2< ... < \hat{\tau}^{(j)}_{\hat{k}^{(j)}} < \hat{\tau}^{(j)}_{\hat{k}^{(j)}+1} = N$. Note that the number of change points $\hat{k}^{(j)}$ does not necessarily equals $k$. It should depend on the way that we encode the observations and the choice of penalty coefficient in (\ref{eq:loss}). So, the observations are partitioned into $\hat{k}^{(j)}+1$ segments and within-segment points are sharing the same estimator of parameter. After that, a vector of length $N$ is generated to record the estimated Bernoulli parameter, denoted as $\{\hat{r}^{(j)}_t\}_{t=1}^N$. Let $\hat{p}^{(j)}_{i}$ be the estimated parameter when $t$ is between $\hat{\tau}^{(j)}_{i-1}$ and $\hat{\tau}^{(j)}_{i}$, so \[
\hat{p}^{(j)}_{i} = \frac{\sum_{t=\hat{\tau}^{(j)}_i+1}^{\hat{\tau}^{(j)}_{i+1}} \mathbbm{1}\{E_t^{(j)}=1\}}{\hat{\tau}^{(j)}_{i+1}-\hat{\tau}^{(j)}_{i}}
\]
for $i=0,1,...,\hat{k}^{(j)}$ and $j=1,2,...,V$. Thus, there are $\hat{\tau}^{(j)}_{i}-\hat{\tau}^{(j)}_{i-1}$ duplicates of $\hat{p}^{(j)}_{i}$ in $\{\hat{r}^{(j)}_t\}_{t=1}^N$, and $\hat{r}^{(j)}_t=\hat{p}^{(j)}_{i}$, for $t \in (\hat{\tau}^{(j)}_{i-1},\hat{\tau}^{(j)}_{i}]$. Repeating the above procedure through the $V$ sequences, we can eventually obtain a sequence of $V$-dimensional estimated parameters, denoted as $\{\hat{r}_t\}_t=\{(\hat{r}^{(1)}_t, \hat{r}^{(2)}_t,..., \hat{r}^{(V)}_t)^{'}\}_t$. 

Generated by marking samples from subarea $B^{(j)}$, the Bernoulli-variable sequence $E^{(j)}_t$ partially reserves the distributional changes from the raw observations. Indeed, the switching pattern of Bernoulli-parameter recorded in $\hat{r}^{(j)}_t$'s is relevant to the distributional changes, while some $\hat{r}^{(j)}_t$'s or at least some subsequences may work as irrelevant noise, especially, when $\int_{B^{(j)}} dF_{i} \cong \int_{B^{(j)}} dF_{i+1}$. An aggregation statistic is present to combine all pieces of information from $j=1,2,...,V$, and weight each $\{E_t^{(j)}\}_{t=1}^{N}$ according to its degree of relevance. In this section, we will treat every sequence equally for theoretical purpose. The weighting procedure will be described in Section~5. 

Different from the CUSUM statistics, we consider the within-group variance in $\{\hat{r}_t\}_t$. Given $k$ hypothesized change point locations $\tau_1,\tau_2,...,\tau_k$, the statistic is written as,
\begin{equation}\label{eq:G}
\hat{G}(\tau_1,\tau_2,...,\tau_k):= \sum_{i=0}^{k} \sum_{t=\tau_i+1}^{\tau_{i+1}} \frac{||\hat{r}_t - \bar{r}_{i}||^2}{\tau_{i+1}-\tau_i}
\end{equation}
where $\bar{r}_{i}={\sum_{t=\tau_i+1}^{\tau_{i+1}}\hat{r}_t}/{(\tau_{i+1}-\tau_i)}$ for $i=0,1,...,k$. Change point locations are then estimated as the ones that minimize the within-group variance, so
\begin{equation}\label{eq:multi}
\hat{\tau}_1,\hat{\tau}_2,...,\hat{\tau}_k = \argmin_{\tau_1,\tau_2,...,\tau_k} \hat{G}(\tau_1,\tau_2,...,\tau_k)
\end{equation}
It is shown in the next section that consistency holds for the statistic. Moreover, it is cheap in the computation when $k>1$. A hierarchical clustering algorithm with $k+1$ clusters obtained can be implemented to search for multiple change point locations.

Stack the estimated parameters $\{\hat{r}_t\}_{t=1}^{N}$ in a $N \times V$design matrix denoted as $\mathcal{M}$, in other words,
\[
\mathcal{M}_{N \times V}=[\mathcal{M}(t,j)]_{t,j}=[\hat{r}_t^{(j)}]_{t,j}~~ \textit{for}~t=1,...,N;~j=1,...,V
\]
A time-order-kept agglomerate hierarchical clustering algorithm is applied upon $\mathcal{M}$ to cluster time locations (rows) with comparable $V$-dimensional covariables. We modify the classical hierarchical clustering algorithm in the sense that only consecutive time points or groups are agglomerated at each iteration, so the original time order is kept. A Wald's type of linkage is applied for the purpose of minimizing the within-group variance. As a result, $k+1$ consecutive time point clusters get returned, so $k$ change point locations can be detected accordingly.

\subsection{Consistency}
We present the consistency of the estimated change point locations obtained from our proposed procedure. It shows that if some of the likelihood-based estimators of a single Bernoulli sequence are consistent, then the estimators derived by the aggregation statistic in (\ref{eq:G}) can also converge the true change point locations. We firstly demonstrate the consistency property in the case of a single change point and then do the same for the multiple change points setting.

Suppose the true change point location is $\tau^*$, so $\{E_t^{(j)}\}_{t=1}^{\tau^*} \sim Bern(p^{(j)}_{1,\tau^*})$ and $\{E_t^{(j)}\}_{t=\tau^{*}+1}^{N} \sim Bern(p^{(j)}_{2,\tau^*})$. By definition,
\begin{equation}\label{eq:r}
\hat{r}^{(j)}_t = \begin{cases} \hat{p}^{(j)}_{1,\hat{\tau}^{(j)}} & ,\, t \in [1,\hat{\tau}^{(j)}]\\
\hat{p}^{(j)}_{2,\hat{\tau}^{(j)}} & ,\, t \in (\hat{\tau}^{(j)},N] \end{cases}
\end{equation}
where $\hat{\tau}^{(j)}$ is the estimated change point locations in $\{E^{(j)}_t\}_t$. To prove the consistency, people typically assume that the size of the two half time sequence cut by $\tau^*$ goes into infinity as $N \rightarrow \infty$, and the proportion of the first half converges to a constant $\gamma^* \in (0,1)$, a.k.a. $\tau^*/N \rightarrow \gamma^* \, (N \rightarrow \infty)$. The within-group variance of (\ref{eq:G}) at any proportion cut $\gamma$ can be written as,
\begin{equation}
\hat{G}(\gamma)= \frac{\sum_{t=1}^{\floor{N\gamma}}||\hat{r}_t - \bar{r}_{1}||^2}{\floor{N\gamma}} + \frac{\sum_{t=\floor{N\gamma}+1}^{N}||\hat{r}_t - \bar{r}_{2}||^2}{N-\floor{N\gamma}}
\end{equation}
where $\bar{r}_{1}=\frac{\sum_{t=1}^{\floor{N\gamma}} \hat{r}_t}{\floor{N\gamma}}$ and $\bar{r}_{2}=\frac{\sum_{t=\floor{N\gamma}+1}^{N} \hat{r}_t}{N-\floor{N\gamma}}$. The estimated change point location now becomes
\begin{equation}
\hat{\tau} = \argmin_{\tau}~\hat{G}(\tau/N)
\end{equation}
in the finite-sample situation.

The theorem below shows that if some of the estimators $\hat{\tau}^{(j)}$ are consistent, then $\hat{\tau}$ consistently converges to $\tau^*$. Assume that if $p_{1,\tau^*}^{(j)} \neq p_{2,\tau^*}^{(j)}$, then $\hat{\tau}^{(j)}/N$ converges to $\gamma^*$ asymptotically; otherwise, $\hat{\tau}^{(j)}/N$ converges to $0$ or $1$ meaning no change point exists in $\{E_t^{(j)}\}_t$. We further assume that there exist at least one encoded Bernoulli sequence such that $p_{1,\tau^*}^{(j)} \neq p_{2,\tau^*}^{(j)}$.
Without loss of the generalization, we suppose that a change point exists in $\{E_t^{(j)}\}_t$ for $j=1,2,...,u$, and no change point exists for $j=(u+1),...,V$ where $1 \le u \le V$.



\begin{theorem}
Under the assumption above and furthermore, for any $\epsilon>0$,
\[
P(|\hat{\tau}/N-\gamma^*| < \epsilon) \to 1
\]
\end{theorem}
as $N \to \infty$.

\begin{proof}
Let $\hat{\gamma}^{(j)}=\hat{\tau}^{(j)}/N$. For any $\gamma \in (0,1)$, rewrite
\[
\hat{G}(\gamma)= \sum_{j=1}^V g(\hat{\gamma}^{(j)},\gamma) (\hat{p}_{1,\hat{\tau}^{(j)}}^{(j)}-\hat{p}_{2,\hat{\tau}^{(j)}}^{(j)})^2
\]
where
$$g(\hat{\gamma}^{(j)},\gamma)= \frac{\hat{\gamma}^{(j)}}{\gamma}(1-\frac{\hat{\gamma}^{(j)}}{\gamma})\mathbbm{1} \{\gamma\ge\hat{\gamma}^{(j)}\} + \frac{1-\hat{\gamma}^{(j)}}{1-\gamma}(1-\frac{1-\hat{\gamma}^{(j)}}{1-\gamma})\mathbbm{1} \{\gamma<\hat{\gamma}^{(j)}\}$$
For $j=1,...,u$, with the consistency of $\hat{\tau}^{(j)}$, we can have
\[
g(\hat{\gamma}^{(j)},\gamma) (\hat{p}_{1,\hat{\tau}^{(j)}}^{(j)}-\hat{p}_{2,\hat{\tau}^{(j)}}^{(j)})^2 \rightarrow g(\gamma^*,\gamma)(p_{1,\tau^*}^{(j)}-p_{2,\tau^*}^{(j)})^2
\]
While for $j=(u+1),...,V$, it shows
\[
g(\hat{\gamma}^{(j)},\gamma) (\hat{p}_{1,\hat{\tau}^{(j)}}^{(j)}-\hat{p}_{2,\hat{\tau}^{(j)}}^{(j)})^2 \rightarrow 0
\]
since $g(0,\gamma)=g(1,\gamma)=0$. Therefore,
\[
\hat{G}(\gamma) \rightarrow \sum_{j=1}^u g(\gamma^*,\gamma)(p_{1,\tau^*}^{(j)}-p_{2,\tau^*}^{(j)})^2 = g(\gamma^*,\gamma) ||p^u_{1,\tau^*}-p^u_{2,\tau^*}||^2= G(\gamma)
\]
as $N \rightarrow \infty$, uniformly in $\gamma$.
Let $\hat{\gamma}=\hat{\tau}/N$. It follows that
\[
\hat{G}(\hat{\gamma}) < \hat{G}(\gamma^*)
\]
Additionally, the minimum value of $g(\gamma^*,\gamma)$ is attained when $\gamma=\gamma^*$. For any $\epsilon>0$, there exists $\eta>0$, such that $G(\gamma)-G(\gamma^*)>\eta$, for all $\gamma$ with $|\gamma-\gamma^*|\ge\epsilon$. Therefore,
\begin{equation*}
\begin{aligned}
P(|\hat{\gamma}-\gamma^*|>\epsilon) &\le P(G(\hat{\gamma})-G(\gamma^*)>\eta) \\
& = P(G(\hat{\gamma})-\hat{G}(\hat{\gamma})+\hat{G}(\hat{\gamma})-G(\gamma^*)>\eta)\\
&\le P(G(\hat{\gamma})-\hat{G}(\hat{\gamma})+\hat{G}(\gamma^*)-G(\gamma^*)>\eta)\\
&\le P(|G(\hat{\gamma})-\hat{G}(\hat{\gamma})| > \eta/2) + P(|\hat{G}(\gamma^*)-G(\gamma^*)| > \eta/2) \to 0
\end{aligned}
\end{equation*}
as $N$ goes into infinity.
\end{proof}

The assumption ensures that $\hat{\tau}^{(j)}$ is a consistent estimator if a change point exists in $\{E^{(j)}_t\}_t$. So long as $u \ge 1$, the distributional discrepancy is captured by $\hat{\tau}$. For a Bernoulli-variable sequence, the change point analysis is relatively easier. One can test the existence of a single change point and plug in a consistent estimator if reject.

In the more general case of multiple change points, suppose that the observations are independent and distributed from $k+1$ distributions $\{F_i\}_{i=0}^k$. Let $\tau^*_i/N \to \gamma^*_i$ as $N\to \infty$, and $0=\gamma^*_0<\gamma^*_1<...<\gamma^*_{k}<\gamma^*_{k+1}=1$. Since $\{E^{(j)}_t\}_t$ may only reserve partial information of the distributional discrepancy, the number of change points in $\{E^{(j)}_t\}_t$ could be smaller than $k$ and varies for different $j$. By further assuming the existence of consistent estimator in the Bernoulli-variable sequence, the theorem below shows the consistency of the aggregation statistic when the number of change point $k>1$.

\begin{theorem}
Define that $C_i = \{j : \hat{\tau}^{(j)}_i/N \to \gamma_i^{*}~as~N\to \infty \}$. Suppose that $|C_i| \ge 1$ and $\hat{\tau}^{(j)}_i$ is none if $j \in \{1,...,V\}/C_i$. Further assume that $\zeta_i + \zeta_{i+1}< \tau^*_{i+1} - \tau^*_i$ where $\zeta_i = \max_{j \in C_i} |\hat{\tau}^{(j)}_i-\tau^{*}_i|$, for $i=1,...,k$. Then, for any $\epsilon>0$,
\[
P(\max_{i=1,...,k} |\hat{\tau}_i/N-\gamma^{*}_i| < \epsilon)\rightarrow 1
\]
as $N \to \infty$.
\end{theorem}
\begin{proof}
Denote $\tau_i^*=N\gamma_i^*$. Consider a group of change point locations that $\Tilde{\tau}_i=\tau_i^*+\zeta_i$, for $i=1,2,...,k$. By the definition of $\zeta_i$, it follows that
\begin{equation*}
\begin{aligned}
\hat{G}(\Tilde{\tau}_1,...,\Tilde{\tau}_k) 
&\le \sum_{i=0}^{k-1} \sum_{j\in C_{i+1}}\frac{2|\hat{\tau}^{(j)}_i-\tau^{*}_i|}{\tau^{*}_{i+1}-\tau^{*}_i} (1-\frac{2|\hat{\tau}^{(j)}_i-\tau^{*}_i|}{\tau^{*}_{i+1}-\tau^{*}_i})(\hat{p}^{(j)}_{i+1}-\hat{p}^{(j)}_{i})^2\\
&\le \sum_{i=0}^{k-1} |C_{i+1}| \frac{2\zeta_i/N}{\gamma^{*}_{i+1}-\gamma^{*}_i} (1-\frac{2\zeta_i/N}{\gamma^{*}_{i+1}-\gamma^{*}_i})
\end{aligned}
\end{equation*}
Then, denote $\Theta= \{(\tau_1,...,\tau_k) : \max_{i=1,...,k} |\tau_i/N-\gamma^{*}_i| \le \epsilon\}$. It shows that, for any $\epsilon>0$,
\[
P(\max_{i=1,...,k} |\hat{\tau}_i/N-\gamma^*_i| \ge \epsilon) \le P(G(\hat{\tau}_1,...,\hat{\tau}_k) \le \min_{(\tau_1,...,\tau_k) \in \Theta} G(\tau_1,...,\tau_k)) \label{eq:star} \tag{*}
\]
Moreover, since $\zeta_i$ is consistent to $0$, uniformly in $i$, by the assumption. So, 
$$P(\max_{i=1,...,k}\zeta_i > \epsilon) \rightarrow 0$$
Therefore,
\begin{equation*}
\begin{aligned}
\eqref{eq:star} \quad
& \le P(G(\hat{\tau}_1,...,\hat{\tau}_k) \le \min_{(\tau_1,...,\tau_k) \in \Theta} G(\tau_1,...,\tau_k) | \max_{i}\zeta_i < \epsilon)P(\max_{i}\zeta_i < \epsilon) \\
&+ P(G(\hat{\tau}_1,...,\hat{\tau}_k) \le \min_{(\tau_1,...,\tau_k) \in \Theta} G(\tau_1,...,\tau_k) | \max_{i}\zeta_i \ge \epsilon)P(\max_{i}\zeta_i \ge \epsilon) \\
&\le P(G(\hat{\tau}_1,...,\hat{\tau}_k) \le G(\Tilde{\tau}_1,...,\Tilde{\tau}_k)) + P(\max_{i}\zeta_i \ge \epsilon) \\
&\rightarrow P(G(\hat{\tau}_1,...,\hat{\tau}_k) \le 0) + 0 = 0
\end{aligned}
\end{equation*}
as $N$ goes into infinity.
\end{proof}

The theorem requires that the estimator is consistent if it exists, and there exists at least one estimator over the $V$ Bernoulli sequences according to a true change point. Though the assumption is strong theoretically, we actually transform the change point detection for unknown underlying distributions into an analysis of a Bernoulli-variable sequence. The task becomes easier since an explicit likelihood function exists without further assumption of the distribution family. So, parametric approaches is involved and fitted under the framework. In practice, the searching algorithm advocated in Section~2 can be employed to detect the change points for each Bernoulli process. Another advantages by applying the encoding-and-aggregation is that the error rate of change point detection can get controlled at the same time, which is present in the next section.

\section{Stability Change Point Analysis}

Finally, it comes to the most general case that the number of change points and their locations are unknown. The current approaches can be divided into two types: model selection and multi-stage testing. A searching algorithm is usually applied in conjunction with a model selection procedure to explore a possible number of change points starting from 1 to a large number. Multi-stage testing is conducted to test the null hypothesis of no additional change point needed by inserting another change point at each stage. However, none of the approaches provides a control for the discovery error of the change point detection. Indeed, the result is sensitive to the objective function of model selection or the significance level in multi-stage testing.

\subsection{The Stability Detection Method}

In this section, we borrow the idea of stability variable selection and propose a robust change point detection framework, named stability detection. Stability selection was firstly advocated by (\cite{Meinshausen and Buhlmann}) to enhance the robustness and control the false discovery rate of variable selection. Half of the samples are randomly selected to feed into a base model at each iteration. The relevant variables are ultimately discovered based on the votes aggregated over all the variable selection results. Later on, (\cite{Beinrucker}) extent the stability selection by sampling disjoint subsets of samples. 

Similar to the strategy of subsampling, we select but not randomly a subset of samples in $B^{(j)}$ to generate a Bernoulli sequence, and then estimate the number and locations of change points within each of the Bernoulli sequences, respectively. By treating each time location as a variable, the stability selection framework can be employed here to aggregate the estimated change points over $B^{(j)}$ for $j=1,2,...,V$. The successive change points are the ones with votes or selected probability above a pre-determined threshold. However, it could be unrealistic to break down the chronological order and treat each time point as separate from others. The locations near the true change points are considered as acceptable results.

Denote that $\mathcal{S}^{(j)}$ is a set of change points detected based on Bernoulli sequence $\{E^{(j)}_t\}_t$, and $p^{(j)}(t)$ is the probability that a time point $t$ is selected, i.e. $p^{(j)}(t)=P(t \in \mathcal{S}^{(j)})$. After aggregating all the change points sets $\mathcal{S}^{(j)}$ for $j=1,2,...,V$, the probability of selection for time point $t$ is defined by
\begin{equation}
\Pi^V(t)=\frac{\sum_{j=1}^V \mathbbm{1}\{t \in \mathcal{S}^{(j)}\}}{V}
\end{equation}
Then, we can obtain the output of stability change point detection by thresholding the quantity with a threshold $\pi \in (0,1)$,
\begin{equation}
\mathcal{S}^{V}_{\pi}=\{t: \Pi^V(t) \ge \pi\}
\end{equation}

\subsection{Error Control}

To evaluate the false discovery rate, we need to define the noisy time points that we should exclude from the admissive set. Especially, we believe that time points around the truth change point $\tau^*$ are admissive, but time points far away from $\tau^*$ should get excluded. 
Define $\mathcal{A}=\{t: t\in (\tau_i^* -w_\mathcal{A},\tau_i^* +w_\mathcal{A})~i=1,2,...\}$ as a set of admissive change points including true change points and their close neighbors. Here, $w_\mathcal{A}$ is an admissive window width and it can change over $i$. Similarly, define $\mathcal{N}=\{t: t \notin (\tau_i^* -w_\mathcal{N},\tau_i^* +w_\mathcal{N})~i=1,2,...\}$ as a set of noisy time points which is outside from the neighbors of the true change points where $w_\mathcal{N}$ is a noisy window width. Note the window width $w_\mathcal{A}$ can be narrower than $w_\mathcal{N}$, such that $\mathcal{A} \subset \mathcal{N}^C$. Suppose that the following assumptions hold for appropriate $w_\mathcal{A}$ and $w_\mathcal{N}$. 
\begin{enumerate}
    \item $\sum_{j=1}^V p^{(j)}(t)/V$ are identical for any $t \in \mathcal{N}$,
    \item $\sum_{j=1}^V p^{(j)}(t)/V$ are identical for any $t \in \mathcal{A}$.
\end{enumerate}
Here, we assume that the noisy time points have the same expected probability to be selected, and so do the admissive time points. Under these assumptions, the next theorem is shown to bound the expectation of false positive rate or false negative rate of change point detection, depending on the choice of threshold $\pi$. 

\begin{theorem}
Under the assumption (1) and (2), denote $p^{V}_{\mathcal{N}} = \sum_{j=1}^V p^{(j)}(t)/V$ for $t \in \mathcal{N}$ and $p^{V}_{\mathcal{A}} = \sum_{j=1}^V p^{(j)}(t)/V$ for $t \in \mathcal{A}$. Let $\pi \in (0,1)$ be the selection threshold.

For any $0 < \xi < 1/p^{V}_{\mathcal{N}} -1$, if $\pi > (1+\xi)p^{V}_{\mathcal{N}}$ we have
\begin{equation}\label{ineq1}
\frac{E[|\mathcal{S}^{V}_{\pi} \cap \mathcal{N}|]}{|\mathcal{N}|} \le \frac{1-(1+\xi)p^{V}_{\mathcal{N}}}{\pi-(1+\xi)p^{V}_{\mathcal{N}}}~exp(- \frac{\xi^2 V}{\xi+2} p^{V}_{\mathcal{N}})
\end{equation}

For any $0 < \xi < 1$, if $\pi < (1-\xi)p^{V}_{\mathcal{A}}$ we have
\begin{equation}\label{ineq2}
\frac{E[|(\mathcal{S}^{V}_{\pi})^C \cap \mathcal{A}|]}{|\mathcal{A}|} \le \frac{(1-\xi)p^{V}_{\mathcal{A}}}{(1-\xi)p^{V}_{\mathcal{A}}-\pi}~exp(- \frac{\xi^2 V}{\xi+2} p^{V}_{\mathcal{A}})
\end{equation}
\end{theorem}
\begin{proof}
For any $0 < \xi < V/\sum_{j=1}^V p^{(j)}(t) -1$, denote $\pi_{\mathcal{N}} = (1+\xi)\sum_{j=1}^V p^{(j)}(t)/V$, so that $\pi_{\mathcal{N}} \in (0,1)$. \\
It is easy to show that $\Pi^V(t) \le (1-\pi_{\mathcal{N}})\mathbbm{1}\{\Pi^V(t) \ge \pi_{\mathcal{N}}\} + \pi_{\mathcal{N}}$ for a fix $t \in \{1,2,...,N\}$. Thus,
\begin{equation*}
\begin{aligned}
    P(\Pi^V(t) \ge \pi) &\le P((1-\pi_{\mathcal{N}})\mathbbm{1}\{\Pi^V(t) \ge \pi_{\mathcal{N}}\} + \pi_{\mathcal{N}} \ge \pi)\\
    & = P(\mathbbm{1}\{\Pi^V(t) \ge \pi_{\mathcal{N}}\} \ge \frac{\pi - \pi_{\mathcal{N}}}{1 - \pi_{\mathcal{N}}})\\
    & \le \frac{1 - \pi_{\mathcal{N}}}{\pi - \pi_{\mathcal{N}}} P(\Pi^V(t) \ge \pi_{\mathcal{N}})\\
    & = \frac{1 - \pi_{\mathcal{N}}}{\pi - \pi_{\mathcal{N}}} P(\sum_{j=1}^V \mathbbm{1}\{t \in \mathcal{S}^{(j)}\} \ge (1+\xi)\sum_{j=1}^V p^{(j)}(t))
\end{aligned}
\end{equation*}
The last inequality holds based on Markov's inequality and the condition that $\pi>\pi_{\mathcal{N}}$.\\
Moreover, $\mathbbm{1}\{t \in \mathcal{S}^{(j)}\}$ are independent for $j=1,2,...,V$. It holds because that we select disjoint samples to make up $\{E^{(j)}_t\}_t$ so for a fixed time $t$, its selection does not reply on the iteration index $j$. The resultant probability can be further bounded via Chernoff upper bound,
\begin{equation*}
\begin{aligned}
    P(\sum_{j=1}^V \mathbbm{1}\{t \in \mathcal{S}^{(j)}\} \ge (1+\xi)\sum_{j=1}^V p^{(j)}(t)) \le exp(-\frac{\xi^2}{\xi+2} \sum_{j=1}^V p^{(j)}(t))
\end{aligned}
\end{equation*}
Hence,
\begin{equation*}
\begin{aligned}
\frac{E[|\mathcal{S}^{V}_{\pi} \cap \mathcal{N}|]}{|\mathcal{N}|} &= \frac{\sum_{t \in \mathcal{N}} P(\Pi^V(t) \ge \pi)}{|\mathcal{N}|}\\
& \le \sum_{t \in \mathcal{N}} \frac{1 - \pi_{\mathcal{N}}}{\pi - \pi_{\mathcal{N}}}~exp(-\frac{\xi^2}{\xi+2} \sum_{j=1}^V p^{(j)}(t))~/~|\mathcal{N}|
\end{aligned}
\end{equation*}
By further assuming identical $\sum_{j=1}^V p^{(j)}(t)$ for $t \in \mathcal{N}$, we can cancel $\mathcal{N}$ for both numerator and denominator, so the inequality (\ref{ineq1}) is obtained. Inequality (\ref{ineq2}) can be proved similarly via the lower bound of Chernoff's.
\end{proof}

As the bound of false positive rate or false negative rate decays with $V$, we are tempting to choose the number of iterations as large as possible. But it will significantly harm the power of change point detection due to the reductive sample size of the recurrent times. In order to control the false discovery rate from both sides, one should increase the signal-selection rate $p^{V}_{\mathcal{A}}$ and decrease the noise-selection rate $p^{V}_{\mathcal{N}}$. It is ideal to set threshold in between, $(1+\xi)p^{V}_{\mathcal{N}} < \pi < (1-\xi)p^{V}_{\mathcal{A}}$. Recall the definition of the selection set $\mathcal{S}^{(j)}:=\{\hat{\tau}^{(j)}_i,~i=1,...,\hat{k}^{(j)}\}$ where $\hat{\tau}^{(j)}_i$ is the $i$-th estimator of the $j$-th Bernoulli sequence. Then, $p^{V}_{\mathcal{A}}$ can be simplified by $\sum_{j=1}^V P(\hat{\tau}^{(j)}_i \in (\tau_i^* -w_\mathcal{A},\tau_i^* +w_\mathcal{A}))/V$. Thus, with a fixed width of $w_\mathcal{A}$, a good estimator $\hat{\tau}^{(j)}_i$ is favored so that it is close the true change point location with a higher probability. 

Another way to increase $p^{V}_{\mathcal{A}}$ is to sightly expand the selection set $\mathcal{S}^{(j)}$, that is to say, selecting the estimators and their neighbors. So, $\mathcal{S}^{(j)}=\{t: t \in neig(\hat{\tau}^{(j)}_i),~i=1,...,\hat{k}^{(j)}\}$. A wider neighbor set $neig()$ is better in adsorbing admissive change points but endures the risk of involving noise. In the change point analysis of a sequence of Bernoulli variables, it is illustrated in Section~2 that a change point is estimated at the locations of 1's. A conservative way of expanding the selection set is to involve the estimator and the locations between the last and the next 1's.

\section{Subsampling and Weighting Strategy}

From an application perspective, there are still two real problems to be addressed. Firstly, how to generate a series of subarea $\{B^{(j)}\}_{j=1,...,V}$ in the encoding phase. Secondly, how to weight the contribution for each encoded Bernoulli sequence $\{E_t^{(j)}\}_t$ based on its degree of relevance. A follow-up question is that how to measure the goodness-of-fit for each $\{E_t^{(j)}\}_t$ and weighting their contributions accordingly. In this section, we resolve both problems via a subsampling weighting technique.

To address the first one, a natural way is to apply clustering analysis to obtain $V$ disjoint clusters as $\{B^{(j)}\}$. But it raises another problem related to the robustness of a different number of clusters and the second question becomes even hard due to the unbalanced cluster size. Model selection criterion in (\ref{eq:loss}) can be used to measure the goodness-of-fit if the cluster sizes are balanced. To ensure robustness and efficiency, we attempt to generate a larger number of clusters but with fixed cluster size, so overlappings are present here. Our numerical experiments show that the method is not sensitive to the choice of $V$. It is advocated to choose a larger number $V=50$ with a fixed subsampling proportion $M/N=0.1$, so a sample is expected to be selected 5 times. 

Denote $\mathbb{X}=[X_1, X_2, ..., X_N]^{'}$ as a $N \times p$ matrix recording the time series $\{X_t\}_{t=1}^N$ where $X_t \in \mathbb{R}^p$. The subsampling algorithm is described as follows. We firstly apply K-Means upon $\mathbb{X}$ to get $V$ cluster centroids. Then, cycle through every centroid to search for its $M$ nearest neighbors in $\mathbb{X}$. We mark the $M$ samples as 1 and the other $N-M$ as 0 at each iteration, so $V$ Bernoulli sequences get returned. If without confusion, let's denote the $M$ marked samples in the $j$-th step as $B^{(j)}$.

Since the weight is inversely proportional to the model selection criterion values or loss in (\ref{eq:loss}), one can consider a mapping function $\mathcal{F}: \mathbb{R}\mapsto \mathbb{R}$ to scale the quantity,
\[
\mathcal{F}(x)=1-\frac{x-\min(x)}{\max(x)-\min(x)}
\]
so, the weight $w^{(j)}$ measuring the importance of $j$-th Bernoulli sequence is defined by,
\begin{equation}\label{eq:w1}
w^{(j)} = \frac{\mathcal{F}(L^{(j)})}{\sum_{j=1}^V \mathcal{F}(L^{(j)})}
\end{equation}
where $L^{(j)} = L(\hat{\tau}^{(j)}_1,\hat{\tau}^{(j)}_2,...)$ is the loss of the $j$-th sequence. Thus, a $N \times V$ weighted design matrix $\mathcal{M}^{weighted}$ is fed into the time-order-kept hierarchical clustering algorithm mentioned in Section~3.3,
\begin{equation}\label{eq:weight_matrix}
\mathcal{M}^{weighted}=\mathcal{M}_{N \times V} \times diag(w^{(1)},...,w^{(V)})
\end{equation}

Another weighting technique is based on the iterative weighting algorithm proposed in (\cite{2}). In the simple case that only one change point exists within a Bernoulli process, it is hard or even impossible to detect the parameter change if the two Bernoulli parameters are too close. Indeed, one can qualify the goodness-of-fit via the difference between $p^{(j)}_{1,\tau^*}$ and $p^{(j)}_{2,\tau^*}$ or the estimated delta $|\hat{p}^{(j)}_{1,\hat{\tau}}-\hat{p}^{(j)}_{2,\hat{\tau}}|$ in practice. If we assume that the size of the two segments is equal, the estimated delta is then simplified by measuring the proportion of the two recovered segments in $B^{(j)}$. The more purity of $B^{(j)}$, the better $E_t^{(j)}$ can be fitted. It enlightens us to measure the Shannon entropy in $B^{(j)}$ as an approximation when $k>1$. 

Denote the weight of the $j$-th sequence at the current step as $w_c^{(j)}$ and the entropy of set $B^{(j)}$ at the current step as $H_c(B^{(j)})$. We can iteratively apply clustering algorithm upon the weighted matrix in (\ref{eq:weight_matrix}) and update the the entropy $H_c(B^{(j)})$ based on the recovered segments in the current step. So, the weight in the next step can be updated by
\begin{equation}\label{eq:wb}
w_{c+1}^{(j)} = 0.5~w_{c}^{(j)} + 0.5~ \frac{\mathcal{F}(H_c(B^{(j)}))}{\sum_{j=1}^V \mathcal{F}(H_c(B^{(j)}))}
\end{equation}
iteratively until convergence. Here the 0.5 is set to smooth the learning curve and make the sum of weights equal 1.

\section{Numerical Experiment}

In this section, we simulate various univariate and multivariate distributions with a known and unknown number of change points to illustrate our model performance. 

When $k$ is known, the time-order-kept hierarchical clustering algorithm is implemented with the weighting techniques proposed in (\ref{eq:w1}) and (\ref{eq:wb}). To differentiate the two weighting techniques, we denote (\ref{eq:w1}) as `simple weighting' and (\ref{eq:wb}) as `iterative weighting'. We compare the performance of our approaches with other nonparametric methods: E-Divisive by (\cite{Matteson and James}), Kernel Multiple Change Point(KernelMCP) by (\cite{Arlot}), and MultiRank by (\cite{Lung-Yut-Fong}). For the fairness of the comparison, all the procedures are conducted with the number of change point $k$ known. The results are reported in Section~6.1 and Section~6.2 for univariate and multivariate settings, respectively. When $k$ is unknown, since it is hard to quantify the false discovery rate, only stability detection is applied in Section~6.3.

Our method was implemented with $V=50$, cluster proportion $M/N=0.1$, and $\phi(N)=2$ (AIC). In the iterative weighting, we further set iteration number $R=150$ and stop criteria when the weights do not change for $10$ steps. E-Divisive was implemented via $ecp$ package with the tuning parameter $\alpha=1$, $R=499$ which was advocated in the paper. KernelMCP was implemented by Python package named $Chapydette$ using the default setting (Gaussian kernel with the Euclidean distance, bandwidth = 0.1, and $\alpha=2$). For MultiRank, we implement R codes provided in the supplementary file of (\cite{Matteson and James}).

To quantify the performance of a change point detection result, we calculate the Adjusted Rand Index(ARI) (\cite{Hubert and Arabie}) between the recovered segments and the true segments. Rand Index(RI) (\cite{Rand}) was originally used to measure the similarity between two data clustering results. As a corrected version of RI, ARI was designed to adjust for the chance of grouping elements. An ARI value of 1 corresponds to a perfect result, while negative or 0 values imply that the recovered segments are different from the underlying segments.

\subsection{Univariate Simulation}

In this section, we simulate univariate distributions with different variance or tailedness. Three data segments are generated sequentially with distributions $\mathcal{N}(0,1)$, $\mathcal{G}$, $\mathcal{N}(0,1)$, respectively. For changes in variance, $\mathcal{G} \sim \mathcal{N}(0,\sigma^2)$; and for changes in tailedness, $\mathcal{G} \sim t_{df}(0,1)$. Unbalanced segments are generated with the sample size $n$, $2n$, $n$, respectively. The size $n$ is also varied $n=100, 200, 300$ while the proportion for the three segments are kept the same.

Given the number of change points, ARI values as recovery accuracy are compared for the proposed approaches, E-Divisive, KernelMCP, and RankMCP in Table\ref{tab1} and Table\ref{tab2}. Results show that E-Divisive outperforms others in the setting of changes in variance. The overall performance is worse in the setting of changes in tailedness, but the iterative weighting approach performs slightly better than others. KernelMCP takes advantage if $\mathcal{G}$ is Gaussian distributed but fails otherwise. It shows that the kernel-type method is very sensitive to the choice of kernel. As a nonparametric approach designed for changes in mean, RankMCP consistently fails in both settings. 

It is remarked that our approach is implemented under unfavorable conditions since we attempt to clustering univariate observations with large cluster numbers. Indeed, the encoding phase can be modified accordingly by applying quantile thresholds and marking extreme observations below or above the thresholds. More discussions about encoding a single-dimensional process are referred to \cite{1}. 

\begin{table}[H]
\caption{ARI values in univariate Gaussian setting}
\label{tab1}
\scalebox{0.9}{
\hspace*{-1cm}
\begin{tabular}{@{}lcccccc@{}}
\hline
& & \multicolumn{5}{c}{univariate distribution with changes in variance} \\
\cline{3-7}
$n$ &$\sigma$ &
\multicolumn{1}{c}{simple weighting} &
\multicolumn{1}{c}{iterative weighting} &
\multicolumn{1}{c}{E-Divisive}&
\multicolumn{1}{c}{KernelMCP}&
\multicolumn{1}{c@{}}{RankMCP} \\
\hline
{100} 
& $1.5$ & 0.4216 (0.1668)&	\textbf{0.5308} (0.1883)&	0.5122 (0.2118)&	0.3146 (0.1888)&	0.3341 (0.1012) \\

& $2$  & 0.6239 (0.1993)&	0.6463 (0.1677)&	\textbf{0.8214} (0.1999)&	0.5248 (0.3274)&	0.3253 (0.0832)  \\
          
& $4$   & 0.8179 (0.1372)&	0.7634 (0.1242)&	\textbf{0.9724} (0.0382)&	0.9510 (0.0784)&	0.3281 (0.0782) \\[6pt]

{200} 
& $1.5$   & 0.5852 (0.2208)&	0.6808 (0.1911)&	\textbf{0.6697} (0.2706)&	0.4355 (0.2601)&	0.3068 (0.1085)  \\

& $2$   & 0.7788 (0.1381)&	0.7653 (0.1367)&	\textbf{0.9536} (0.0605)&	0.8942 (0.1554)&	0.3208 (0.0878) \\

& $4$   & 0.9184 (0.0425)&	0.8821 (0.0850)&	\textbf{0.9872} (0.0176)&	0.9815 (0.0157)&	0.3246 (0.0798) \\[6pt]

{300} 
& $1.5$   & 0.7311 (0.1747)&	0.7674 (0.1586)&	\textbf{0.7905} (0.2393)&	0.6048 (0.3100)&	0.3429 (0.0933) \\

& $2$   & 0.8375 (0.1019)&	0.8189 (0.1054)&	\textbf{0.9758} (0.0274)&	0.9610 (0.0338)&	0.3449 (0.0861) \\
          
& $4$   & 0.9440 (0.0311)&	0.9243 (0.0480)&	\textbf{0.9935} (0.0078)&	0.9889 (0.0088)&	0.3449 (0.0861) \\
\hline
\end{tabular}
}
\end{table}

\begin{table}[H]
\caption{ARI values in univariate student-t setting}
\label{tab2}
\scalebox{0.9}{
\hspace*{-1cm}
\begin{tabular}{@{}lcccccc@{}}
\hline
& & \multicolumn{5}{c}{univariate distribution with changes in tailedness} \\
\cline{3-7}
$n$ &$df$ &
\multicolumn{1}{c}{simple weighting} &
\multicolumn{1}{c}{iterative weighting} &
\multicolumn{1}{c}{E-Divisive}&
\multicolumn{1}{c}{KernelMCP}&
\multicolumn{1}{c@{}}{RankMCP} \\
\hline
{100} 
& $1$ & 0.5527 (0.2242)&	0.6426 (0.1780)&	\textbf{0.6880} (0.2533)&	0.2764 (0.1630)	& 0.3297 (0.1110) \\

& $2$  & 0.3742 (0.1582)&	\textbf{0.4930} (0.1980)	&	0.4542 (0.1851)&	0.2930 (0.1488)&	0.2954 (0.0944)  \\
          
& $5$   & 0.3045 (0.1380)&	\textbf{0.3910} (0.1316)&	0.3767 (0.1139)&	0.2564 (0.1380)&	0.3194 (0.1247) \\[6pt]

{200} 
& $1$   & 0.7695 (0.1585)&	0.7762 (0.1407)&	\textbf{0.8419} (0.2123)&	0.3391 (0.2089)&	0.3205 (0.0879)  \\

& $2$   & 0.4332 (0.2181)&	\textbf{0.6097} (0.1824)&	0.5055 (0.2267)&	0.2672 (0.1648)&	0.3212 (0.0967) \\

& $5$   & 0.2887 (0.1254)&	\textbf{0.3696} (0.1666)&	0.3606 (0.1204)&	0.2401 (0.1600)&	0.2899 (0.1280) \\[6pt]

{300} 
& $1$   & 0.8395 (0.0926)&	0.8220 (0.1065)&	\textbf{0.8927} (0.1709)&	0.4675 (0.2891)&	0.3325 (0.1010) \\

& $2$   & 0.5010 (0.2481)&	\textbf{0.6605} (0.2129)&	0.6547 (0.2564)&	0.3026 (0.1965)&	0.3288 (0.1191) \\
          
& $5$   & 0.3101 (0.1211)&	\textbf{0.4263} (0.1713)&	0.3484 (0.1237)&	0.2655 (0.1494)&	0.2890 (0.1166) \\
\hline
\end{tabular}
}
\end{table}

\subsection{Multivariate Simulation}

Following the generation step above, we simulate multivariate observations in this section. The observations are distributed from $\mathcal{N}_d(0,I)$, $\mathcal{N}_d(0,\Sigma)$, $\mathcal{N}_d(0,I)$, respectively. In the first part, we consider binormal distributions, in which $\Sigma=\begin{bmatrix}
1 & \rho \\ \rho & 1 \end{bmatrix}$ with different correlation $\rho$. The ARI values for E-Divisive and KernelMCP are compared in Table\ref{tab3}. Given a moderate $\rho$ value, it shows that the weighting procedures have comparable ARI values and outperform E-Divisive and KernelMCP. When $\rho$ is extremely large and the sample size is greater, the binormal distribution actually degrades to an univariate Gaussian, which explains why the ARIs of E-Divisive and KernelMCP come from behind at $\rho=0.9$ and $n=300$.

In the second part, we simulate observations with dimension $d=3,5,10$. Since KernelMCP is not easily adoptive when the dimension is more than 2, we only compare the performance of simple weighting, iterative weighting, and E-Divisive. Two types of $\Sigma$ are imposed for the generation. $\Sigma_1$ is set with diagonal elements $1$ and off-diagonal elements $\rho$; $\Sigma_2$ is set with diagonal elements $1$ and $\pm$1-off-diagonal elements $\rho$. Table\ref{tab4} shows that the simple weighting is more favorable in identifying the change point locations in the case of $\Sigma_1$; in the more complicated case of $\Sigma_2$, the iterative weighting performs the best.

\begin{table}[H]
\caption{ARI values in 2-dim Gaussian setting}
\label{tab3}
\begin{tabular}{@{}lccccc@{}}
\hline
& & \multicolumn{4}{c}{2-dim Gaussian with changes in correlation} \\
\cline{3-6}
$n$ &$\rho$ &
\multicolumn{1}{c}{simple weighting} &
\multicolumn{1}{c}{iterative weighting} &
\multicolumn{1}{c}{E-Divisive}&
\multicolumn{1}{c@{}}{KernelMCP} \\
\hline
{100} 
& $0.5$ & 0.3673 (0.1593) &	\textbf{0.4671} (0.1531) &	0.3904 (0.1303) & 0.2862 (0.1458) \\

& $0.7$  & 0.4726 (0.1935) &	\textbf{0.5542} (0.1862) &	0.4309 (0.1695) &	0.2877 (0.1411)  \\
          
& $0.9$   & \textbf{0.6993} (0.1882) &	0.6612 (0.1803)	& 0.5985 (0.2550) &	0.3453 (0.2032) \\[6pt]

{200} 
& $0.5$   & 0.3982 (0.1830)	& \textbf{0.5386} (0.1950) &	0.3960 (0.1592) &	0.2776 (0.1533)  \\

& $0.7$   & \textbf{0.7186} (0.1790) &	0.6535 (0.1912) &	0.5025 (0.2340)	 & 0.2944 (0.1494) \\

& $0.9$   & 0.8316 (0.1435)	& 0.7665 (0.1409) &	\textbf{0.8614} (0.2108) &	0.6924 (0.2816) \\[6pt]

{300} 
& $0.5$   & 0.5171 (0.2343) &	\textbf{0.5809} (0.2228)	&	0.3897 (0.1643)	& 0.2872 (0.1435) \\

& $0.7$   & \textbf{0.8209} (0.1099)	& 0.7311 (0.1558) &	0.7013 (0.2852) &	0.3196 (0.1411) \\
          
& $0.9$   & 0.8753 (0.1188)	& 0.8160 (0.1233)	&	\textbf{0.9461} (0.1403) &	0.9305 (0.1404) \\
\hline
\end{tabular}
\end{table}

\begin{table}[H]
\caption{ARI values in d-dim Gaussian setting}
\label{tab4}
\scalebox{0.8}{
\hspace*{-2cm}
\begin{tabular}{@{}lcccccccc@{}}
\hline
& & \multicolumn{3}{c}{d-dim Gaussian with off-diagonal correlation 0.5} & & \multicolumn{3}{c}{d-dim Gaussian with $\pm$1-off-diagonal correlation 0.5}\\
\cline{3-5} \cline{7-9}
$n$ &$d$ &
\multicolumn{1}{c}{simple weighting} &
\multicolumn{1}{c}{iterative weighting} &
\multicolumn{1}{c}{E-Divisive} &
&
\multicolumn{1}{c}{simple weighting} &
\multicolumn{1}{c}{iterative weighting}&
\multicolumn{1}{c@{}}{E-Divisive} \\
\hline
{100} 
& $3$ & 0.4074 (0.1965)&	\textbf{0.5520} (0.1855)& 0.4505 (0.1772) && 0.3805 (0.1565)&	\textbf{0.5276} (0.2014)&	0.4367 (0.1673) \\

& $5$  & 0.5660 (0.2095)&	\textbf{0.6364} (0.1760)& 0.4704 (0.2018)&& 0.4056 (0.1904)&	\textbf{0.5160} (0.1836)	&	0.4222 (0.1581) \\
          
& $10$   & \textbf{0.7826} (0.1595)&	0.7132 (0.1785)&	0.5777 (0.2563)& & 0.3919 (0.1426)&	\textbf{0.5362} (0.1913)&	0.4306 (0.1641)\\[6pt]

{200} 
& $3$   & 0.5056 (0.2296)&	\textbf{0.6500} (0.2027)&	0.4281 (0.1976)&& 0.5013 (0.2267)&	\textbf{0.5899} (0.1859)&	0.4041 (0.1802) \\

& $5$   & \textbf{0.8078} (0.1312)&	0.7706 (0.1649)&0.6051 (0.2646)&& 0.4594 (0.2090)&	\textbf{0.5977} (0.2101)&	0.4419 (0.1849) \\

& $10$   & \textbf{0.8819}(0.0957)&	0.8386 (0.1255)&	0.7875 (0.2644)&&0.4148 (0.1936)&	\textbf{0.6229} (0.1983)&	0.4459 (0.1905)\\[6pt]

{300} 
& $3$   & 0.6150 (0.2422)&	\textbf{0.7480} (0.1828)&0.5166 (0.2285)&&0.6258 (0.2461)&	\textbf{0.6319} (0.2108)&	0.4780 (0.2199) \\

& $5$   & \textbf{0.8649} (0.0796)&	0.8322 (0.1089)&	0.8007 (0.2601)&&0.5690 (0.2563)&	\textbf{0.6836} (0.2066)&	0.5429 (0.2343)\\
          
& $10$   & 0.8973 (0.0744)&	0.8679 (0.1059)&\textbf{0.9170} (0.1787)&& 0.5586 (0.2469)&	\textbf{0.7149} (0.1784)&	0.4969 (0.2224)\\
\hline
\end{tabular}
}
\end{table}

\subsection{Simulation for Stability Detection}

Stability detection is applied when the number of change points is unknown. Suppose we encode the continuous observations into $V$ Bernoulli sequence components, so there are $V$ change point sets obtained in total. Instead of weighting each voting set equally, we involve the simple weighting techniques and weight the votes according to the goodness-of-fit.

Denote a binormal distribution $\mathcal{N}_2(0,\begin{bmatrix}
1 & 0.7 \\ 0.7 & 1 \end{bmatrix} )$ as $\mathcal{G}_2$. In the first scenario, 3 binormal distributions are generated by $\mathcal{N}_2(0,I)$, $\mathcal{G}_2$, $\mathcal{N}_2(0,I)$ with sample size 300, 600, 300, respectively. In the change point analysis of each Bernoulli sequence, model selection criterion is applied with a different penalty coefficient $\phi(N)$. $\phi(N)$ is set from 2, a value corresponding to AIC, to $log(N)$, a value corresponding to BIC. We partition the time axis into disjoint time bins with the same length. The probability of selection is then calculated based on the summed votes falling in each bin. It shows that the results are not sensitive to the choice of penalty term. There are 6 or 7 curves always above the others regardless of the penalty coefficient. Their corresponding time bins are marked in Figure\ref{fig.penalty}. The first and second bins are (320,360) and (360, 400) which are close to the first change point located at 300; the third and fourth time bins (280,320) and (880,920) cover the true change point locations. The probability of selection for all the time points are shown in Figure\ref{fig.aic_bic}((A) for AIC and (B) for BIC). The two big spikes indicate that the number of change points is 2. By further setting a threshold at 0.1, we can obtain two consecutive time windows containing the truth change point locations.

\begin{figure}[h!]
\centering
\includegraphics[width=5.0in]{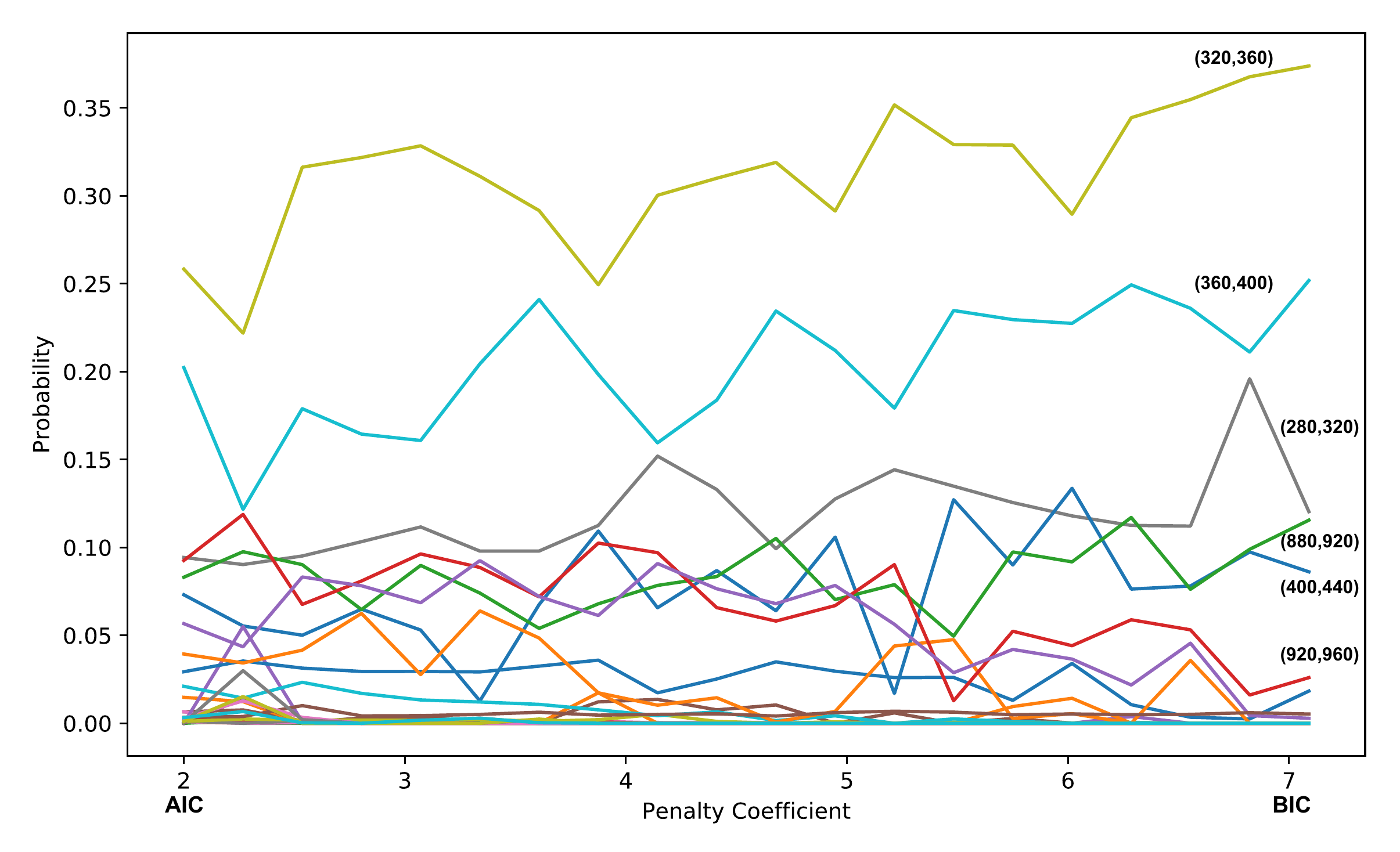}
\caption{probability of selection with different penalty coefficient $\phi(N)$; different time bins are plotted in different curves}
\label{fig.penalty}
\end{figure}

\begin{figure}[h!]
\centering
\includegraphics[width=5.0in]{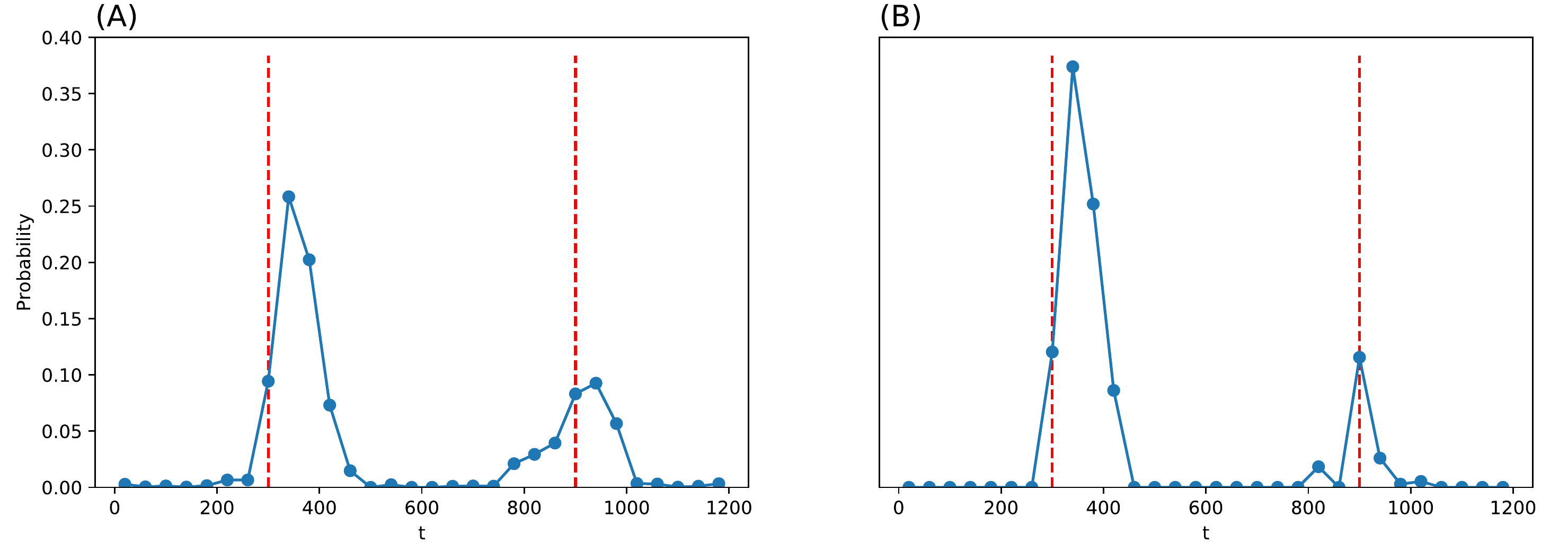}
\caption{(A) probability of selection with $\phi(N)=2$ as AIC; (B) probability of selection with $\phi(N)=log(N)$ as BIC. True change point locations are plotted in vertical lines}
\label{fig.aic_bic}
\end{figure}

In the second scenario, we make the problem more complicated by generating 7 segments $\mathcal{N}_2(0,I)$, $\mathcal{G}_2$, $\mathcal{N}_2(0,I)$, $\mathcal{G}_2$, $\mathcal{N}_2(0,I)$, $\mathcal{G}_2$, $\mathcal{N}_2(0,I)$ with equal sample size $n$. BIC is applied for model selection of each Bernoulli sequence. There turns to be 6 obvious spikes after smoothing the probability curve. Especially, when $n$ increases to 500, the local maximas can almost perfectly detect the true change points.

Indeed, stability detection gives a new perspective to discover the number of change points and measure the confidence bound simultaneously. To estimate the change point locations, one can search for the top$k$ local maxima or go back to the hierarchical clustering procedure with an estimated change point number $k$. Based on our experiment, the two estimation results are very similar. We claim the consistency should also hold based on the result in Figure\ref{fig.simulation}.

\begin{figure}[h!]
\centering
\includegraphics[width=5.0in]{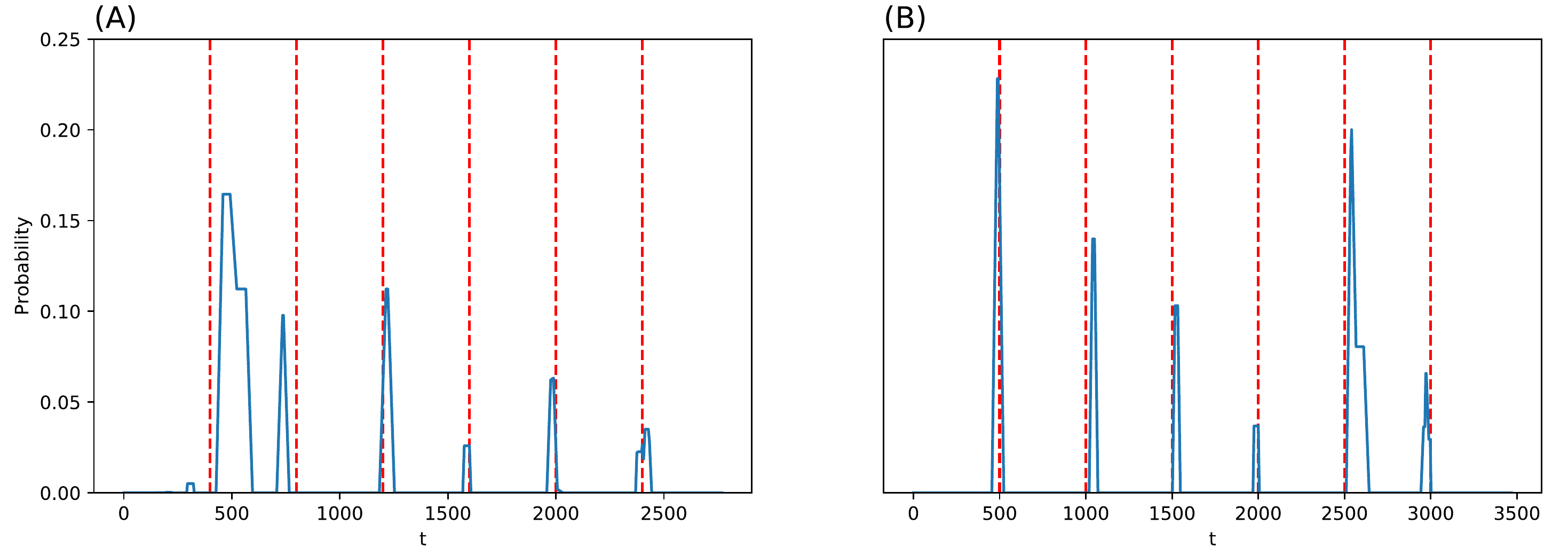}
\caption{(A) probability of selection with $n=400$; (B) probability of selection with $n=500$}
\label{fig.simulation}
\end{figure}

\section{Real Data Application}

\subsection{Genome Data}

CpG dinucleotide clusters or ‘CpG
islands’ are genome subsequences with a relatively high number of CG dinucleotides (a cytosine followed by a guanine). They are observed close to transcription start sites (\cite{Sxonov}) and play a crucial role in gene expression regulation and cell differentiation (\cite{Bird}). There were developed many computational tools for CpG island identification. A sliding window is typically employed to scan the genome sequence to figure out CpG islands based on some filtering criteria. However, the criteria are set with subjective choice (G$+$c proportion, observation versus expectation ratio, etc) and it has evolved over time. It commonly happens that different CpG island finders would provide various results. 

In this section, we implement our change point detection approach in the categorical nucleotide sequence. We believe that the proposed algorithm is able to detect an abrupt change in C-G patterns, and the estimated change point locations may help researchers to identify potential CpG islands. A contig (accession number $\textit{NT\_000874.1}$) on human
chromosome 19 was taken as an example for CpG island searching. 
The dataset is available on the website of National Center for Biotechnology Information(NCBI). 

Denote the genome sequence as $\{X_t\}_{t=1}^N$ with $X_t \in \{A,G,T,C\}$. In the encoding phase, a 0-1 sequence $\{E_t\}_t$ is generated such that $E_t=1$ if $X_{t}=C \& X_{t+1}=G$ and $E_t=0$ otherwise, for $t=1,...,N-1$. Algorithm1 is implemented to search for multiple change points in the Bernoulli sequence.
Results from a CpG island searching software CpGIE (\cite{Wang and Leung}) are shown as a benchmark for comparison. Criteria advocated by the authors are employed in the usage of CpGIE (length $\ge 500$ bp, G $+$ C content $\ge 50\%$ and CpG O/E ratio $\ge 0.60$). Note that our algorithm does not need any assumption or tuning parameter. The result in Figure$\ref{fig.DNA}$ shows that there is a high proportion of overlapping segments between ours and CpGIE's. Our approach can also find extra genome subsequence with a higher number of C-Gs which are misspecified by CpGIE.

\begin{figure}[h!]
\centering
\includegraphics[width=5in]{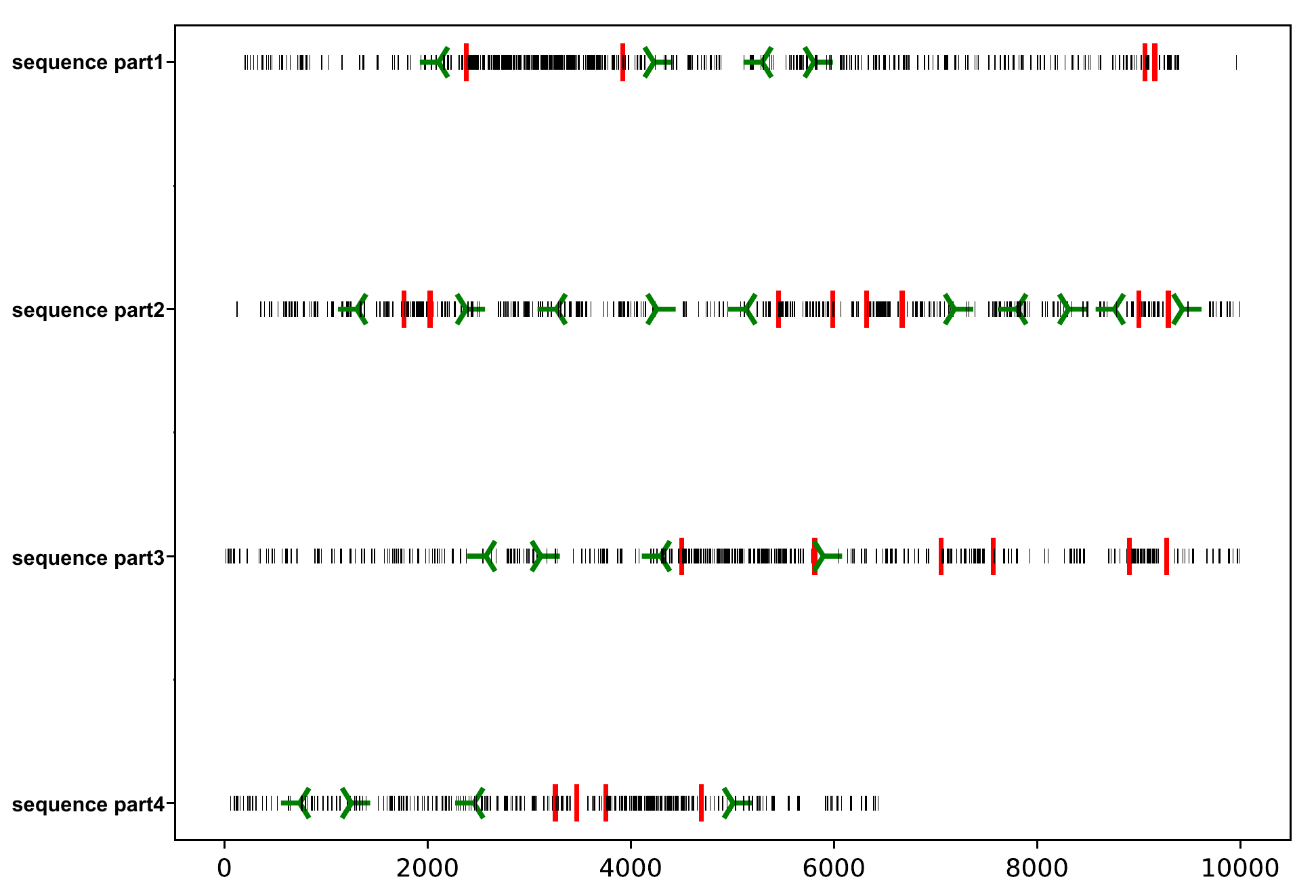}
\caption{encoded DNA sequence-CG dinucleotides are marked in black; the CpG islands discovered by CpGIE are marked in green; the estimated change point locations are marked in red}
\label{fig.DNA}
\end{figure}

\subsection{Hurricane Data}

It was widely recognized that the global temperature has risen due to anthropogenic factors, such as increased carbon dioxide emissions and other human activities. According to NOAA's 2020 global climate report, the annual temperature has increased globally at an average rate of 0.14 degrees Fahrenheit per decade since 1880 and over twice that rate (0.32 degrees Fahrenheit) since 1981. It was argued by climatologists that the warmer sea surface leads to an increasing number of stronger tropical cyclones (\cite{Emanuell}; \cite{Saunders and Lee}). However, (\cite{Landsea}) believes that the warmer sear surface increases only weak cyclones which are short and even hard to be detected. In this section, we studied the number of cyclones between 1851 and 2019. We are interested to detect potential change points embedded within the tropical cyclone history.

The dataset HURDAT2 recording the activities of cyclones in the Atlantic basin is available on the website of National Oceanic Center(NHC). NHC tracked the intensity of each tropical cyclone per 6 hours every day (at 0, 6, 12, and 18). The intensity level is categorized based on wind strength in knots, such as hurricane (intensity greater than 64 knots), tropical storm (intensity between 34 and 63 knots), tropical depression (intensity less than 34 knots). Different from (\cite{Robbins}) in categorizing cyclones, we summarize the number of time units that a category is observed, so the count is at most $4\times31$ in a month. The monthly frequency of tropical storm-level and higher-level cyclones is reported in Figure\ref{fig.hurdat}(A). If we apply 5 change points which is detected by the local maxima of stability detection in Figure\ref{fig.hurdat}(B), the time range is then partitioned based on the variation of storm count. Figure\ref{fig.hurdat}(A) shows that storms are more active in the 1880s, 1960s and after 2000. Though the global temperature trends to go upward since 1980, the storms are relatively sparse between 1980 and 2000. Thus, we tend to believe that no firm conclusion can be made yet that higher temperatures would increase the number of hurricanes.

\begin{figure}[h!]
\centering
\includegraphics[width=5.0in]{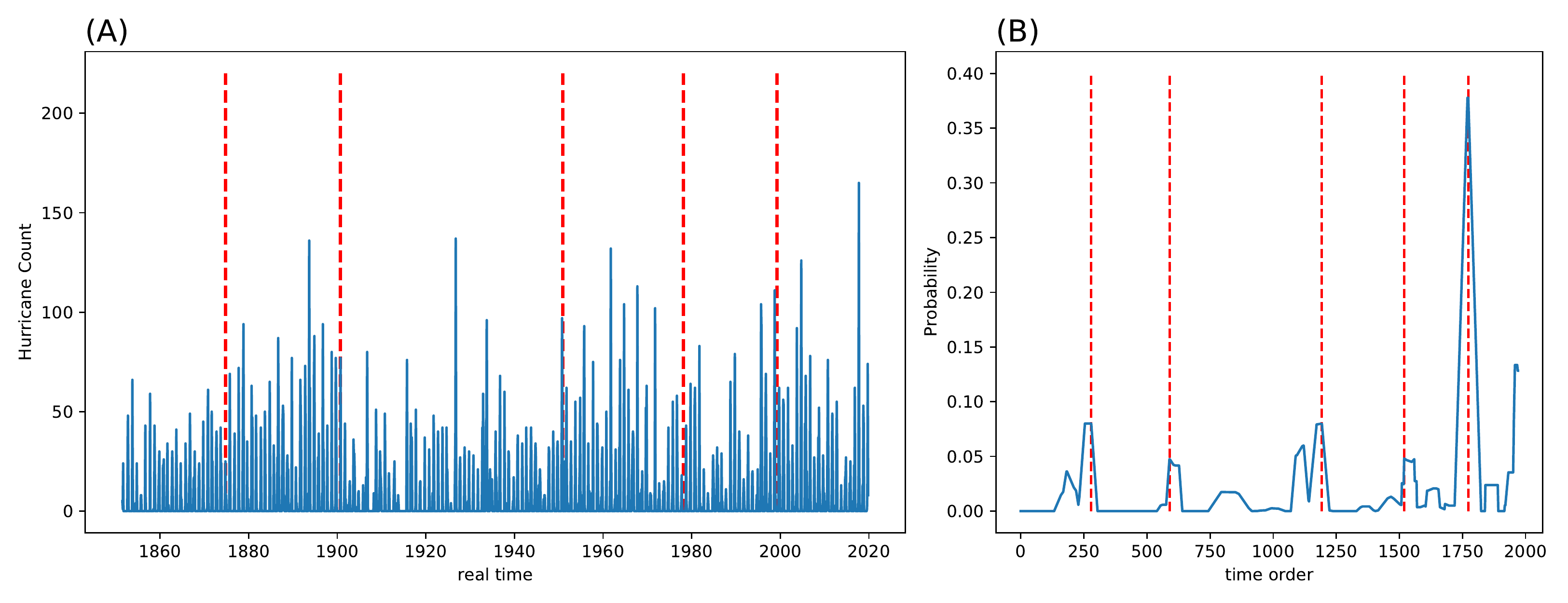}
\caption{(A) monthly hurricane counts in Atlantic basin from year 1851 to 2019; estimated change points are plotted in vertical lines. (B) probability of selection for all the time points; local maximas are plotted in vertical lines}
\label{fig.hurdat}
\end{figure}

\subsection{Financial Data}

Lastly, the proposed approach is applied to detect the abrupt time-varying dependence within bivariate stock log returns. CTSH and IBM are chosen as representative of IT Consulting subcategories of S\&P500 based on Global Industrial Classification Standard(GICS). The first and last hours in the transaction time are filtered out (so it is from 10am to 4pm), and the hourly price returns are calculated in the business days of the year 2006. A constant is added to the returns of CTSH for a better visualization in Figure$\ref{fig.pairwise}$(A), but the raw return series are analyzed. It was noted that the lagged correlation statistics are not significant based on the sample autocorrelation function of stock returns. Conditional heteroskedasticity can be studied by a more complicated time series model, like GARCH, but it is out of our concentration.

We encode the bivariate time series and apply stability detection techniques. Figure$\ref{fig.pairwise}$(B) shows that there exist 3 or 4 change points within the returns. The top3 change point locations with the highest probability are marked by vertical lines in Figure$\ref{fig.pairwise}$(A). It shows that the returns are partitioned into segments with different volatility levels. If we further look into the scatterplot between CTSH and IBM under different time partitions (left, middle, right segments) in Figure$\ref{fig.scatterplot}$, both returns in the middle phase are relatively high, and their correlation is even stronger.

\begin{figure}[h!]
\centering
\includegraphics[width=5.0in]{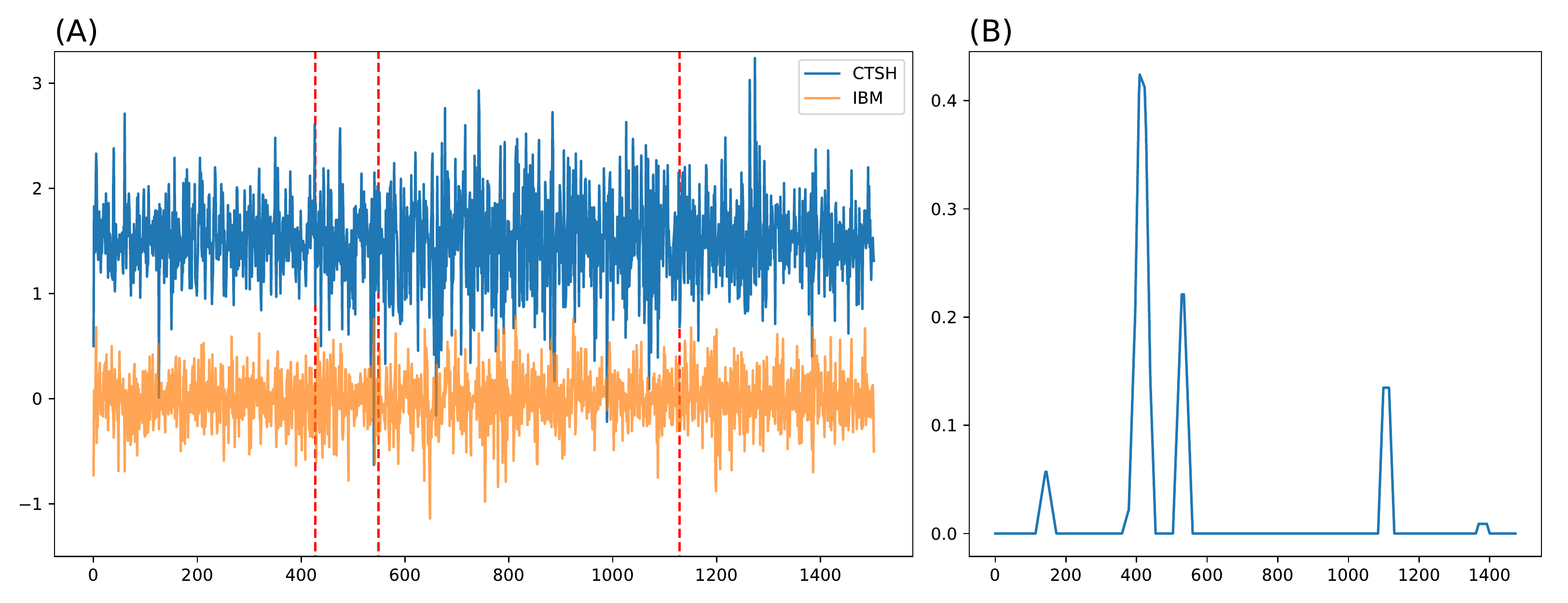}
\caption{(A) hourly index returns of CTSH and IBM in 2006; top3 change points with the highest probability of selection are plotted in vertical lines. (B) probability of selection for all time points}
\label{fig.pairwise}
\end{figure}

\begin{figure}[h!]
\centering
\includegraphics[width=5.0in]{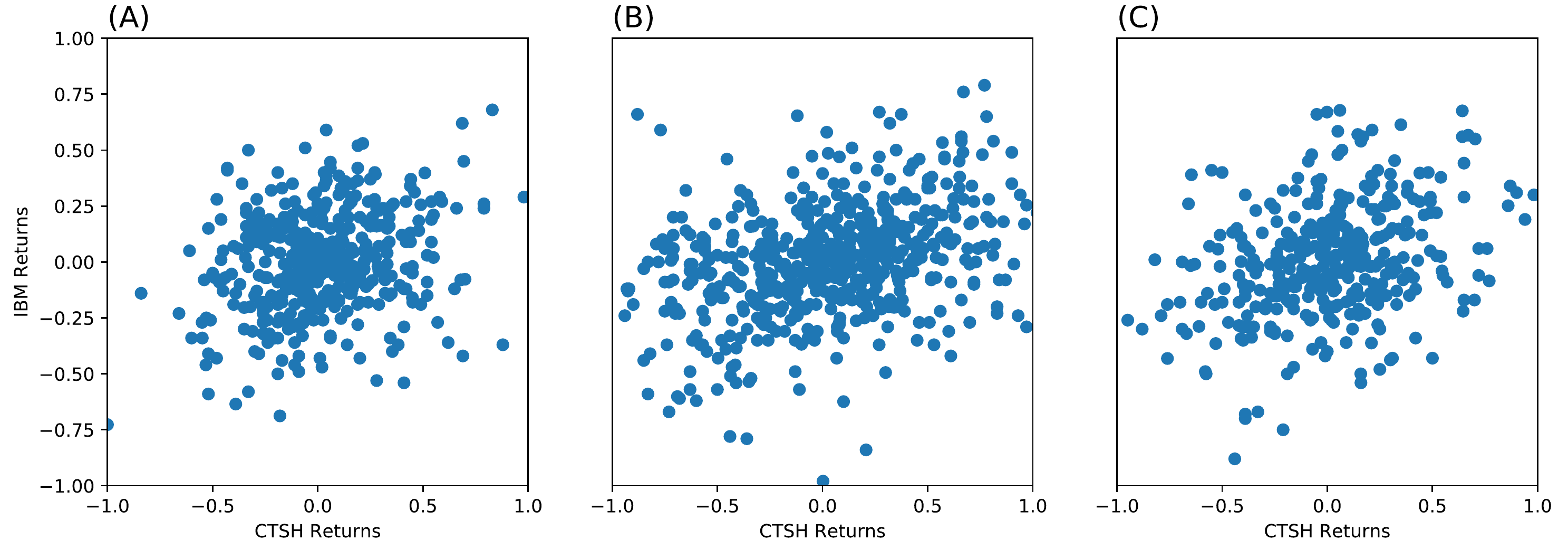}
\caption{scatterplot of returns of CTSH versus IBM; (A) observations on the left segment; (B) observations on the middle segment; (C) observations on the right segment}
\label{fig.scatterplot}
\end{figure}

\newpage

\section*{Conclusion}

In the paper, we have established a framework to encode a sequence of continuous observations into several Bernoulli processes and proposed approaches for change point detection in univariate and multivariate settings with or without a known number of change points. Theoretical work shows that the proposed method can hold both asymptotic property and finite-sample error control. Numerical and real experiments show that the approach is able to detect any type of distributional changes and can be applied to categorical, ordinal, and continuous data. Furthermore, the computational expense is reasonable with time complexity at the most expensive part $O(VM^2)$ or $O(VN^2)$, and parallel programming is applicable to decrease the complexity to $O(N^2)$.

\end{document}